\documentclass[runningheads,a4paper]{llncs}

\usepackage{amssymb,amsmath}

\usepackage{graphicx}
\usepackage{color}
\usepackage{colortbl}
\usepackage{bm}
\usepackage{complexity}
\usepackage{framed}
\usepackage{enumitem}
\usepackage[caption=false]{subfig}
\usepackage{float}
\usepackage{hyperref}
\usepackage{cite}
\usepackage{wasysym}
\usepackage[colorinlistoftodos,prependcaption,textsize=tiny]{todonotes}
\usepackage[latin1]{inputenc}
\usepackage{soul}
\usepackage{slashbox}

\newtheorem{innercustomgeneric}{\customgenericname}
\providecommand{\customgenericname}{}
\newcommand{\newcustomtheorem}[2]{%
  \newenvironment{#1}[1]
  {%
   \renewcommand\customgenericname{#2}%
   \renewcommand\theinnercustomgeneric{##1}%
   \innercustomgeneric
  }
  {\endinnercustomgeneric}
}

\newcustomtheorem{customprop}{Proposition}
\newcustomtheorem{customlmm}{Lemma}
\newcustomtheorem{customdef}{Definition}
\newcustomtheorem{customthm}{Theorem}

\newcounter{theonew4}

\newtheorem{def2}[theonew4]{Definition}

\usepackage[english,algoruled,longend,ruled,vlined,linesnumbered]{algorithm2e}

\def\qed{{\hfill\hbox{\rlap{$\sqcap$}$\sqcup$}}}

\usepackage{url}

\newcommand{\keywords}[1]{\par\addvspace\baselineskip
\noindent\keywordname\enspace\ignorespaces#1}

\begin{document}

\mainmatter  
\title{Total tessellation cover and quantum walk
\thanks{This work was partially supported by the Brazilian agencies CAPES, CNPq and FAPERJ.}}

\titlerunning{\textit{Total tessellation cover and quantum walk}}

\author{A.~Abreu\inst{1}
\and L.~Cunha\inst{1}\and C.~de Figueiredo\inst{1}\\ F.~Marquezino\inst{1}\and D.~Posner\inst{1}\and R.~Portugal\inst{2}}
\authorrunning{A. Abreu et al.}
\institute{Universidade Federal do Rio de Janeiro -- \email{\{santiago, lfignacio, celina, franklin, posner\}@cos.ufrj.br} \and Laborat\'orio Nacional de Computa\c{c}\~ao Cient\'{\i}fica -- \email{portugal@lncc.br}}

\toctitle{Lecture Notes in Computer Science}
\tocauthor{Alexandre S. de Abreu, Lu\'is Felipe I. Cunha, Celina M.H. de Figueiredo, Franklin de L. Marquezino, Daniel F.D. Posner, Renato Portugal}
\maketitle

\begin{abstract}

We propose the total staggered quantum walk model and the total tessellation cover of a graph. This model uses the concept of total tessellation cover to describe the motion of the walker who is allowed to hop both to vertices and edges of the graph, in contrast with previous models in which the walker hops either to vertices or edges.
We establish bounds on $T_t(G)$, which is the smallest number of tessellations required in a total tessellation cover of $G$. 
We highlight two of these lower bounds $T_t(G) \geq \omega(G)$ and $T_t(G)\geq is(G)+1$, where $\omega(G)$ is the size of a maximum clique and $is(G)$ is the number of edges of a maximum induced star subgraph.
Using these bounds, we define the good total tessellable graphs with either $T_t(G)=\omega(G)$ or $T_t(G)=is(G)+1$.
The \textsc{$k$-total tessellability} problem aims to decide whether a given graph $G$ has $T_t(G) \leq k$.
We show that \textsc{$k$-total tessellability} is in $\mathcal{P}$ for good total tessellable graphs.
We establish the $\mathcal{NP}$-completeness of the following problems when restricted to the following classes: \textsc{($is(G)+1$)-total tessellability} for graphs with $\omega(G) = 2$; \textsc{$\omega(G)$-total tessellability} for graphs $G$ with $is(G)+1 = 3$; \textsc{$k$-total tessellability} for graphs $G$ with $\max\{\omega(G), is(G)+1\}$ far from $k$; and \textsc{$4$-total tessellability} for graphs $G$ with $\omega(G) = is(G)+1 = 4$. As a consequence, we establish hardness results for bipartite graphs, line graphs of triangle-free graphs, universal graphs, planar graphs, and $(2,1)$-chordal graphs.


\end{abstract}
\section{Introduction}\label{sec:one}

A \textit{tessellation} of a graph $G=(V,E)$ is a partition of $V$ into vertex disjoint cliques called \textit{tiles}. 
A \textit{$k$-tessellation cover} of $G$ is a set of $k$ tessellations that covers $E$. The \emph{tessellation cover number} $T(G)$ of a graph $G$ is the size of a minimum tessellation cover. 
The \textsc{$k$-tessellability} problem aims to decide whether a given graph $G$ has $T(G) \leq k$. The concept of tessellations on graphs was introduced in~\cite{PSFG16}.
See \cite{BkWest} for basic definitions and notations in graph theory.

\begin{def2}
\emph{Let $G = (V, E)$ be a graph and $\Sigma$ a non-empty label set. 
A \emph{total tessellation cover} comprises a proper vertex coloring and a tessellation cover of $G$ both with labels in $\Sigma$ such that, for any vertex $v\in V$, there is no edge $e\in E$ incident to $v$ so that $e$ belongs to a tessellation with label equal to the color of~$v$.}
\end{def2}


An alternative way to characterize a tessellation is by describing the edges that belong to the tessellation. 
A $k$-tessellation cover of $G = (V, E)$ is a function $h$ that assigns to each edge of $E$ a nonempty subset in $\mathcal{P}(\Sigma)$, where 
$\Sigma = \{1, \ldots, k\}$, such that the set of edges having the same label corresponds to a tessellation, i.e., induces a partition of $V$ into cliques.  
A \textit{$k$-total tessellation cover} of a graph $G$ simultaneously assigns labels in $\Sigma$
to $V$ as a proper vertex coloring $f$ and labels in $\mathcal{P}(\Sigma)\setminus \emptyset$ to $E$ as a tessellation cover with function $h$, such that each $uv \in E$ satisfies $f(u) \not\in h(uv)$ and $f(v) \not\in h(uv)$.

\begin{def2}
\emph{The \emph{total tessellation cover number} $T_t(G)$ of a graph $G$ is the minimum size of the set of labels $\Sigma$ for which $G$ has a total tessellation cover. 
The \textsc{$k$-total tessellability} problem aims to decide whether a given graph $G$ has $T_t(G) \leq k$.}
\end{def2}

\subsubsection*{Motivation.}
The quantum computation paradigm has gained popularity due to the recent advances in the physical implementation and in the development of quantum algorithms. There is an important concept, known as quantum walk, which is the mathematical modeling of a walk of a particle on a graph. This concept provides a powerful tool in the development of quantum algorithms~\cite{portugal2013quantum}. Indeed, in the last decades the interest in quantum walks has grown considerably since quantum algorithms that outperform their classical counterparts employ quantum walks~\cite{Ven12, ambainis2004quantum}. 
In 2016, Portugal et al. proposed the staggered quantum walk model~\cite{PSFG16}, which is more general than the previous quantum walk models~\cite{Por16} by containing the Szegedy model~\cite{Szegedy:2004} and part of the flip-flop coined model~\cite{ portugal2013quantum}. 
The staggered quantum walk employs the concept of graph tessellation cover to obtain local unitary matrices such that their product results in the evolution operator for the quantum walk. There is a recipe to obtain a local unitary matrix from a tessellation. The staggered model requires at least two tessellations (corresponding to $2$-tessellable graphs). In a tessellation, each clique establishes a neighborhood around which the walker can move under the action of the associated local unitary matrix. To define the evolution operator, one has to check whether the set of tessellations contains the whole edge set of the graph, since an uncovered edge would play no role in a quantum walk~\cite{PSFG16}. 



\subsubsection*{Related works.}

Abreu et al.~\cite{ArLatin, ArTCS1}   
proved that 
$\chi'(G)$ and 
$\chi(K(G))$ are upper bounds for $T(G)$, where $K(G)$ is the clique graph of $G$. 
They also proved the hardness of \textsc{$k$-tessellability} for planar graphs, $(2,1)$-chordal graphs, and $(1,2)$-graphs and showed that \textsc{$2$-tessellability} is solved in linear time. Since $T(G) = \chi'(G)$ for triangle-free graphs, \textsc{$k$-tessellability} is hard for this graph class~\cite{Kor97}.
Posner et al.~\cite{ArCNMAC} showed that \textsc{$k$-tessellalbility} is $\mathcal{NP}$-complete for line graphs of triangle-free graphs.
Abreu et al.~\cite{ArTCSGood} proved that $is(G)$ is a lower bound for $T(G)$, where $is(G)$ is the number of edges in a maximum induced star of a graph $G$. They prove the hardness of \textsc{$k$-tessellability} for universal graphs and the hardness of \textsc{good tessellable recognition}, which aims to decide whether $G$ is \emph{good tessellable}, i.e.,  $T(G) = is(G)$. 
 The concept of minimum tessellation cover was independently proposed as equivalence dimension in Duchet~\cite{duchet1979representations}, and the relation between the two concepts was described in~\cite{ArTCSGood}.

\subsubsection*{Contributions.}


This work is presented in the following sections. 
Section~\ref{sec:three} contains a study on the bounds of the value of $T_t(G)$. Such bounds describe not only the number of operators required for a total staggered quantum walk model, but they also provide tools to analyse the computational complexity of \textsc{$k$-total tessellability}, which is done in Section~\ref{sec:four}. Since $T_t(G) = \chi_t(G)$ for triangle-free graphs, the problem is hard even when restricted to bipartite graphs~\cite{mcdiarmid1994total}.
We show that \textsc{$k$-total tessellability} is in $\mathcal{P}$ for good total tessellable graphs, and as a by product \textsc{$k$-tessellability} is in $\mathcal{P}$ for good tessellable graphs. On the other hand, we show hardness results for the \textsc{$k$-total tessellability} problem for line graphs of triangle-free graphs, universal graphs, planar graphs, and $(2,1)$-chordal graphs. As a consequence, the \textsc{good total tessellability recognition} problem is $\mathcal{NP}$-complete. Note that there are few results about the hardness of \textsc{total colorability}. 
Section~\ref{sec:two} describes the total staggered quantum walk model, 
which drives a walker to hop both to vertices and edges. It also contains a description of the simulation of the total staggered quantum walk on a graph $G$ in terms of a staggered quantum walk on the total graph of $G$.
In Section~\ref{sec:five}, Table~\ref{tab:concluding} presents the behavior analysis of the computational complexity related to the following parameters: $\chi'(G), \chi_t(G), T(G)$, and $T_t(G)$.

\section{Bounds on $T_t(G)$}\label{sec:three}



Since a total coloring of a graph $G$ induces a total tessellation cover,
\begin{equation}\label{eq:bound2} 
T_t(G) \leq \chi_t(G). 
\end{equation}
Particularly, for triangle-free graphs $T_t(G)\!=\!\chi_t(G)$ because the set of edges in each tessellation of any total tessellation cover is a matching. Hence, \linebreak\textsc{($\Delta\!+\!1$)-total tessellability} is hard even when restricted to regular bipartite graphs~\cite{mcdiarmid1994total}. 
Furthermore, by definition, 
\begin{equation}\label{eq:bound1}
\max\{\chi(G), T(G)\} \leq T_t(G) \leq \chi(G)+T(G).
\end{equation}
Note that the lower bound of Eq.~(\ref{eq:bound1}) implies that $T_t(G)\geq \omega(G)$.


\begin{lemma}  If $\chi(G) \geq 3T(G)$, then $T_t(G) = \chi(G).$
\label{eq:nova1}
\end{lemma}
\begin{proof}
Let $f$ be a proper vertex coloring and $\mathcal{C} = \{\mathcal{T}_1, \mathcal{T}_2, \ldots, \mathcal{T}_{T(G)}\}$ be a $T(G)$-tessellation cover for $G$. We define $\mathcal{C}'$ a tessellation cover for $G$ with $3T(G)$ labels such that $\mathcal{C}'$ is compatible with $f$ as follows. Each tessellation $\mathcal{T}'_i, 1 \leq i \leq 3T(G)$, of $\mathcal{C}'$ is associated with a color $i$. Since $\chi(G) \geq 3T(G)$ there are enough colors. 

The edges of tessellations $\mathcal{T}'_{3j-2}, \mathcal{T}'_{3j-1}$, and $\mathcal{T}'_{3j}$ are given by the edges of the tessellation $\mathcal{T}_j, 1 \leq j \leq T(G)$, such that $\mathcal{T}'_{3j-2}$ (resp. $\mathcal{T}'_{3j-1}$, $\mathcal{T}'_{3j}$) consists of the edges of $\mathcal{T}_j$ that do not have an endpoint with color $3j-2$ (resp. $3j-1, 3j$).~\qed 
\end{proof}

Using an argument similar to the one in the proof of Lemma~\ref{eq:nova1}, we can rewrite the upper bound of Eq.~(\ref{eq:bound1}) as follows  \begin{equation}T_t(G) \leq \max\left\{\chi(G), T(G)+\left\lceil{2\chi(G)}/{3}\right\rceil\right\}.
\label{eq:nova2}
\end{equation}
Eq.(\ref{eq:nova2}) says that $\chi(G) \geq 3T(G)$ implies $T_t(G) = \chi(G)$, or $\chi(G)\leq 3T(G)$ implies $T(G)\leq T_t(G)\leq 3T(G)$. 
In case $\chi(G)=3$, Eq.~(\ref{eq:nova2}) implies that $T(G) \leq T_t(G)\leq T(G)+2$.
An example of a graph $G$ for which $T_t(G)=3T(G)~-~1$ and $T_t(G)>\chi(G)$ has $V(G) = \{v_1, v_2, v_3, v_4\}\cup\{u_1, u_2, u_3, u_4\}\cup\{w_1, w_2, w_3, w_4\}$, where $\{v_1, v_2, v_3, v_4\}$ and $\{u_1, u_2, u_3, u_4\}$ are maximal cliques and $\{v_i, u_i, w_i\}$  are triangles for $1\le i\le 4$. In this case $T_t(G)=5$, $\chi(G)=4$ and $T(G)=2$.
Note that $T_t(G) = \chi(G) + T(G)$ requires that $\chi(G)\leq 2$, i.e., $G$ is bipartite, which implies $T_t(G) = \chi_t(G)$ and $T_t(G)$ may assume only two values: $T_t(G) = \chi(G)+T(G) = \Delta(G) + 2$ or $T_t(G) = \chi(G)+T(G)-1=\Delta(G)+1$.

\begin{lemma}\label{eq:bound3}  
$\!\displaystyle T_t(G)\!\geq\!\!\max_{v \in V(G)}\{\chi(G^c[N(v)])\}\!+\!1\!\geq\!\!\max_{v \in V(G)}\{\omega(G^c[N(v)])\}\!+\!1\!=\!is(G)\!+\!1.$
\end{lemma}
\begin{proof}

Consider a total tessellation cover of a graph $G$, a vertex $v$ of $G$, and $G^c[N(v)]$, which is the complement graph of the graph induced by the neighborhood of $v$. 
In any tessellation, the endpoints of the edges that are incident to $v$ and belong to the tessellation induce a clique, hence the vertices of this clique are a stable set in $G^c[N(v)]$.
Therefore, the tessellations with edges incident to a vertex $v$ induce a vertex coloring of $G^c[N(v)]$, and the number of these tessellations is at least $\chi(G^c[N(v)])$. 
Moreover, these tessellations have labels that are different from the color of vertex $v$. Therefore, $T_t(G) \geq \chi(G^c[N(v)])+1$.
Note that $is(G[N[v]]) = \alpha(G[N(v)]) = \omega(G^c[N(v)])$ and $\displaystyle is(G)=\max_{v \in V(G)}{is(G[N[v]])}$.~\qed
\end{proof}

Graphs with $T_t(G)=T(G)=k$ have no induced subgraph $K_{1,k}$ because  $T_t(G)\geq is(G)+1 \geq k+1$.
Moreover, there is no tile of size $k$ in any tessellation of a total tessellation cover.
If $T_t(G)=T(G)=3$, then $G$ is $K_{1,3}$-free and there is no clique of size three in any tessellation. Therefore, the total tessellation cover of $G$ induces a total coloring of $G$, and the only graphs for which $T_t(G)=T(G)=3$ are the odd cycles with $n$ vertices such that $n\equiv 0 \mod\ 3$.
For bipartite graphs, $T(G)=\Delta(G)$ and $T_t(G) > T(G)$.
For triangle-free graphs, $T_t(G)=T(G)$ if $\chi'(G)=\chi_t(G)=\Delta+1$.
It follows that deciding whether $T_t(G)=T(G)=\Delta(G)+1$ is $\mathcal{NP}$-complete from the proof  that \textsc{($\Delta+1$)-total colorability} is $\mathcal{NP}$-complete for triangle-free snarks~\cite{ArSnark}, which are graphs with  $\chi'(G)=\Delta+1$.


\section{Good Total Tessellable Graphs}
\label{sec:four}

Since the concept of good tessellable graphs introduced in~\cite{ArTCSGood} has provided
keen insights into the hardness of finding minimum-sized tessellation covers, we define the concept of good total tessellable graphs in
order to further explore hardness results related to total tessellation covers. In the quantum computation context, we are interested in graph classes which use as few color labels as possible because the number of operators is as low as possible. In this case, $T_t(G)$ must be close to the lower bounds.

\begin{def2}
\emph{A graph $G$ is \textit{good total tessellable} if either $T_t(G) =\omega(G)$ or $T_t(G)=is(G)+1$. We say that $G$ is \textit{Type I} (resp. \textit{Type II}) if $T_t(G)=\omega(G)$ (resp. $T_t(G) = is(G)+1$).}
\end{def2}

Now we show that \textsc{$k$-total tessellability} is in $\mathcal{P}$ if we know beforehand that the graph is either good total tessellable Type~I or Type~II. 

The \textit{Lov\'{a}sz number} $\vartheta(G)$ is a real number such that $\omega(G^c)\leq \vartheta(G)\leq \chi(G^c)$~\cite{grotschel1981ellipsoid}.  
We denote $\psi(G)$ the integer nearest to $\vartheta(G)$. The value of $\psi(G)$ can be be determined in polynomial time ~\cite{grotschel1981ellipsoid}.

For Type I graphs, $T_t(G)=\omega(G)$. Since Eq.~(\ref{eq:bound1}) implies that $\omega(G) \leq \chi(G)\leq T_t(G)$, we have $\omega(G)=\chi(G)=T_t(G)=\psi(G^c)$.

For Type II graphs, $T_t(G) = is(G)+1$.
For any vertex $v \in V(G)$,
$\omega(G^c[N(v)]) \leq \psi(G[N(v)]) \leq \chi(G^c[N(v)])$, and by Lemma~\ref{eq:bound3}, $T_t(G)\geq\psi(G[N(v)])+1$. 
Since $T_t(G)=is(G)+1$, by Lemma~\ref{eq:bound3} there is a vertex $u \in V(G)$ such that $T_t(G)=\omega(G^c[N(u)])+1$. 
In this case, $\omega(G^c[N(u)])+1 = \chi(G^c[N(u)])+1$, and we determine $\omega(G^c[N(u)])$ using $\psi(G^c[N(u)])$.
Therefore, $\displaystyle T_t(G)\!=\!\!\max_{v \in V(G)}\{\psi(G[N(v)])\}\!+\!1$.

The same method used to determine $T_t(G)$ for Type II graphs can be applied for good tessellable graphs in order to determine $T(G)$, where $\displaystyle T(G)=\max_{v \in V(G)}\{\psi(G[N(v)])\}$. 

\subsubsection{Hardness results.}

As presented in Section~\ref{sec:three}, $(\Delta+1)$-\textsc{total tessellability} is $\mathcal{NP}$-complete for bipartite graphs, which have $is(G)+1=\Delta+1$ and $\omega(G)=2$.
Now, we show that \textsc{$k$-total tessellability} is $\mathcal{NP}$-complete for the following cases: line graph of triangle-free graphs with $k=\omega(G) \geq 9$ and $is(G)+1=3$; universal graphs with $k$ very far apart from both $is(G)+1$ and $\omega(G)$; planar graphs with $k=4=\omega(G)=is(G)+1$; and $(2,1)$-chordal graphs with $k=is(G)+1=\omega(G)+3$. 




\subsubsection*{Line graph of triangle-free graphs.}


Machado et al.~\cite{machado2010chromatic} proved that $k$-\textsc{edge colorability} is $\mathcal{NP}$-complete for $3$-colorable $k$-regular triangle-free graphs if $k \geq 3$.
The key idea of the proof of Theorem~\ref{teo:line} is to verify that $T_t(L(G))=\chi'(G)$ when $k\geq 9$.
The edges incident to any vertex $v$ of graph $G$ correspond to a clique of $L(G)$, whose size is the degree of $v$. 
If two vertices of $G$ are non-adjacent, then the corresponding cliques in $L(G)$ share no vertices. 
Hence, we cover the edges of the cliques of $L(G)$ incident to the vertices of each of the three color class of the $3$-coloring of $G$ with a tessellation related to the color class because these cliques share no vertices.
Therefore, 
since $T(L(G))=3$ and $\chi(L(G)) \geq 9 \geq 3T(G)$, by Lemma~\ref{eq:nova1}, $T_t(L(G))=\chi(L(G))=\chi'(G)$.
Note that in this case $k=\omega(L(G))$ and $is(L(G))+1=3$.

\begin{theorem}
\label{teo:line}
\textsc{$k$-total tessellability} is $\mathcal{NP}$-complete for line graphs $L(G)$ of $3$-colorable $k$-regular triangle-free graphs $G$ for any $k\geq 9$.
\end{theorem}

\subsubsection*{Universal graphs.}

Abreu et al.~\cite{ArTCSGood} reduced $q$-\textsc{colorability} to $k$-\textsc{tessellabi\\lity} for universal graphs. We present a similar argument to establish the $\mathcal{NP}$-completeness of  \textsc{$k$-total tessellability} for universal graphs.  Let $G$ be an instance of $q$-\textsc{colorability}.
The key idea of the proof of Theorem~\ref{teo:univer} is to add to $G^c$ a universal vertex $u$ and $2|V(G)|$ pendant vertices adjacent to $u$, which defines the graph [$2|V(G)|, G^c$] of $G^c$. Now, the total tessellation cover number of the constructed graph is given by $2|V(G)|+\chi(G)+1$, using  
labels $1, \ldots, \chi(G)$ to cover the edges incident to $u$ that belong to the subgraph induced by $V(G^c\cup \{u\})$, labels $\chi(G)+1,\ldots, \chi(G)+2|V(G)|$ to cover the edges incident to the pendant vertices and labels $\chi(G)+1,\ldots, \chi(G)+|V(G)|$ are enough to cover the edges of $G^c$; assign to $u$ the color $2|V(G)| +\chi(G)+1$, to the pendant vertices color $1$, and to the remaining vertices colors  $\chi(G)+|V(G)|+1, \ldots, \chi(G)+2|V(G)|$.  
The minimality follows from Lemma~\ref{eq:bound3}. 
Therefore,  $T_t([2|V(G)|,G^c]) =  2|V(G)|+\chi(G)+1$.

Note that $is(C_5 \vee \{u\})=2$, $T_t(C_5 \vee \{u\})=4$, and any  minimum total tessellation cover of $C_5 \vee \{u\}$ has at least three labels assigned to the edges incident to $u$ and a fourth label assigned to $u$.
Thus, $T_t([2|V(G)|,G^c\cup C_5])=T_t([2|V(G)|,G^c])+3$;  $is([2|V(G)|,G^c\cup C_5])=is([2|V(G)|,G^c])+2$; and\linebreak   $\omega([2|V(G)|,G^c\cup C_5])=\omega([2|V(G)|,G^c])$.
Therefore, each addition of a $C_5$ increases the gap between the total tessellation cover number and both the sizes of a maximum induced star and a maximum clique.
As long as the number of the $C_5$'s is polynomially bounded by the size of $G$, \textsc{$k$-total tessellability} is $\mathcal{NP}$-complete even if $k$ is far apart from $is(G)$ and $\omega(G)$.

\begin{theorem}
\label{teo:univer}
\textsc{$k$-total tessellability} is $\mathcal{NP}$-complete for universal graphs.
\end{theorem}
\subsubsection*{Planar graphs.}

We show that \textsc{$4$-total tessellability} is $\mathcal{NP}$-complete when restricted to planar graphs $G$ with $is(G)+1=\omega(G)=4$. 
We present a polynomial transformation from \textsc{$3$-colorability} when restricted to planar graphs with maximum degree four~\cite{ArGarey} to \textsc{$4$-total tessellability} for planar graphs.
Let $G$ be an instance of such coloring problem.
$G'= G\vee\{u\}$ has a $4$-coloring if and only if the planar graph $G$ has a $3$-coloring.
We define three gadgets as depicted in Fig.~\ref{fig:gadgets}.
The edges of the external triangles of the \textit{Duplicator Gadget} are tiles of size three in a same tessellation.
The edges of the external triangles of the \textit{NotEqual Gadget} are tiles of size three in different tessellations.
The \textit{Shifter Gadget} forces triangles $T_1$ and $T_4$ to be tiles on a tessellation $a$, and triangles $T_2$ and $T_3$ to be tiles on a tessellation $b$ different from $a$.

\begin{figure}
\centering
     \includegraphics[scale=0.25]{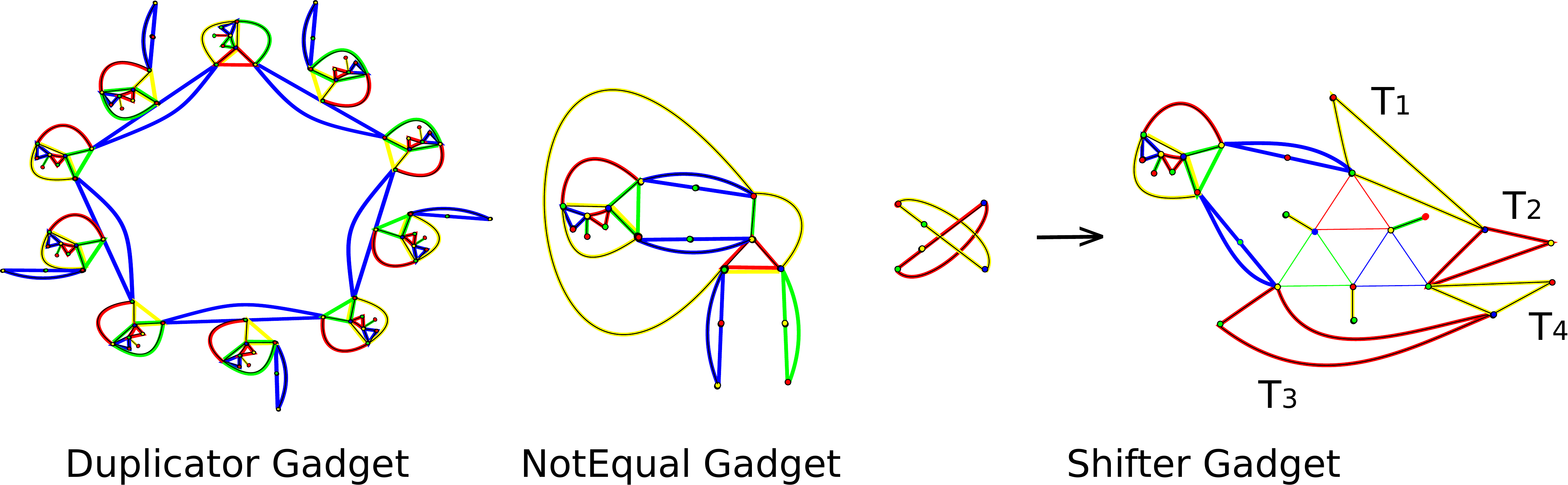}
     \caption{Duplicator Gadget, NotEqual Gadget, and Shifter Gadget.\label{fig:gadgets}}
\end{figure}

Each vertex $v$ of $G'$ is associated with a Duplicator Gadget such that the number of external triangles of the Duplicator Gadget of $v$ is equal to $d_{G'}(v)$. 
If two vertices of $G'$ are adjacent, we connect one external triangle of each Duplicator Gadget with a NotEqual Gadget. 
Thus, in a $4$-total tessellation cover of the obtained graph $H$, the labels of the external triangles of the Duplicator Gadget associated with a vertex $v$ are equal to the color of $v$ in a $4$-coloring of $G'$.
Now, we transform $H$ into a planar graph $H'$ by replacing each crossing triangles of $H$ by a Shifter Gadget.
Therefore, the planar graph $H'$ has a $4$-total tessellation cover if and only if $G$ has a $3$-coloring.
Note that in this case $k=\omega(G)=4=is(G)+1$.

\begin{theorem}
\textsc{$4$-total tessellability} is $\mathcal{NP}$-complete for planar graphs.
\label{teo:planar}
\end{theorem}

\subsubsection*{$(2,1)$-chordal graphs.}

A graph $G$ is $(2,1)$ if its vertex set can be partitioned into two stable sets and one clique.
Since \textsc{$3$-edge colorability} is $\mathcal{NP}$-complete for $3$-regular graphs~\cite{machado2010chromatic}, \textsc{$3$-vertex colorability} for $4$-regular line graphs is also $\mathcal{NP}$-complete.
Let $G$ be a $4$-regular line graph.
We construct a graph $H$ from $G$ as follows. $V(H)$ contains a clique $\{e_0, \ldots, e_{|E(G)|-1}\}$ where each $e_i$, $0 \leq i \leq |E(G)|-1$, is associated with a distinct edge of $G$. $V(H)$ contains an stable set $\{e_0', \ldots, e_{|E(G)|-1}'\}$ such that each $e_i'$ is adjacent to all $e_j$ with $j\neq i$ and $j \neq i+1 \mod\ |E(G)|$.
$V(H)$ contains an stable set $\{v_0, \ldots, v_{|V(G)|-1}\}$, where each $v_i$, $0 \leq i \leq |V(G)|-1$, is associated with a distinct vertex of $G$. 
Each $e_j\in\{e_0, \ldots, e_{|E(G)|-1}\}$ is adjacent to vertices $v_r, v_s\in \{v_0, \ldots, v_{|V(G)|-1}\}$ such that $e_j = v_r v_s$.
$V(H)$ contains an stable set $P$ comprising $(|V(G)|+|E(G)|)(|E(G)|+1)$ pendant vertices such that each vertex of $\{v_0, \ldots, v_{|V(G)|-1}\}\cup\{e'_0, \ldots, e'_{|E(G)|-1}\}$ is adjacent to $|E(G)|+1$ pendant vertices.
By construction, $H$ is $(2,1)$ and chordal.

 We claim that $T_t(H)=|E(G)|+3$  if and only if $\chi(G)=3$. 
 Consider a $3$-coloring $c$ of $G$.
 Obtain a $k$-total tessellation cover of $H$ with $k=|E(G)|+3$ as follows.
 Assign colors in $\{1,\! \ldots,\! |E(G)|\}$ to the vertices of the clique $\{e_0,\! \ldots,\! e_{|E(G)|-1}\!\}$.
 Assign to vertex $e'_i$, for $1 \leq i \leq |E(G)|$, the same color of the vertex $e_i$.
 For  $0 \leq i \leq |E(G)|-1$, the tile with vertices $\{e'_i\} \cup
\{e_j\ |\ j \neq i \mbox{ and } j \neq i+1 \mod\ |E(G)|\}$ is in the
tessellation with label $i+2 \mod |E(G)|$. 
 Note that if two vertices $v_i$ and $v_k$ of $G$ are not adjacent, then the cliques $\{v_i\} \cup \{e_j\ |\ v_i \mbox{ is endpoint of } e_j \mbox{ in } G\}$ and $\{v_k\} \cup \{e_j\ |\ v_k \mbox{ is endpoint of } e_j \mbox{in } G\}$ are disjoint. 
 Thus, the tile with vertices $\{v_i\} \cup \{e_j\ |\ v_i \mbox{ is endpoint of } e_j \mbox{ in } G\}$ is in the tessellation with label $c(v_i)+|E(G)|$. 
 Finally, greedily assign colors and labels to the remaining vertices and edges of $H$.
 Consider a total tessellation cover of $H$ with $k=|E(G)|+3$ labels.
 Note that we require $|E(G)|$ tessellations to cover the edges between the vertices $\{e_0, \ldots, e_{|E(G)|-1}\} \cup \{e'_0, \ldots, e'_{|E(G)|-1}\}$ in any total tessellation cover of $H$. Moreover, a tile in each of those $|E(G)|$ tessellations contains $|E(G)|-2$ vertices of the clique $\{e_0, \ldots, e_{|E(G)|-1}\}$.
 Since each tile $\{v_i\} \cup \{e_j\ |\ v_i \mbox{ is endpoint of } e_j \mbox{ in } G\}$, for $0 \leq i \leq |V(G)|-1$, contains four vertices of the clique $\{e_0, \ldots, e_{|E(G)|-1}\}$, there are only three  tessellation labels used by the tiles $\{v_i\} \cup \{e_j\ |\ v_i \mbox{ is endpoint of } e_j \mbox{ in } G\}$, for $0 \leq i \leq |V(G)|-1$.
 Moreover, if two vertices $v_i$ and $v_k$ are adjacent in $G$, then the tiles $\{v_i\} \cup \{e_j\ |\ v_i \mbox{ is endpoint of } e_j \mbox{ in } G\}$ and $\{v_k\} \cup \{e_j\ |\ v_k \mbox{ is endpoint of } e_j\mbox{ in } G\}$ share a vertex $e_j=v_iv_k$ in $H$ and they  are tiles belonging to different tessellations.
 Hence, we obtain a $3$-coloring $c$ of $G$ as follows.
 Assign the label of the tile $\{v_i\} \cup \{e_j\ |\ v_i \mbox{ is endpoint of } e_j \mbox{ in } G\}$ to the color of $v_i$ in $c$.

Therefore, $G$ has a $3$-coloring if and only if $H$ has a total tessellation cover with $|E(G)|+3$ labels. Note that  $k=is(H)+1=\omega(H)+3=|E(G)|+3$.

\begin{theorem}
\textsc{$k$-total tessellability} is $\mathcal{NP}$-complete for chordal graphs.
\label{teo:chordal}
\end{theorem}

\section{The total staggered quantum walk model}\label{sec:two}

We now show how to simulate a total staggered quantum walk on a graph $G$ with a staggered quantum walk on its total graph $\text{Tot}(G)$.
The \textit{total graph} $\text{Tot}(G)$ of $G$ has $V(\text{Tot}(G))=V(G) \cup E(G)$ and $E(\text{Tot}(G))= E(G) \cup \{u\ uw\ |\ u \in V(G) \mbox{, } uw \in E(G)\} \cup \{uv\ vw\ |\ uv \in E(G) \mbox{ and } vw \in E(G)\}$.
Let $A=\text{Tot}(G)$, $A[E(G)] = Y$ and $A[V(G)] = X$. Subgraph $Y$ is isomorphic to the line graph $L(G)$ of $G$, and $X$ is isomorphic to the original $G$. 
We define the clique $K_v = \{v\} \cup \{vw\ |\ vw \in E(G)\}$ of $A$.

Consider a total tessellation cover of a graph $G$.
Define an associated tessellation cover of $A$ as follows. 
Assign the labels of the edges of $G$ to the respective edges of $X$ and assign the color of each vertex $v$ of $G$ to the edges of $A[K_v]$. 
We simulate the total staggered quantum walk on $G$ with the staggered quantum walk on $A$ by considering the vertices of $G$ as the corresponding vertices of $X$ in $A$, and the edges of $G$ as the corresponding vertices of $Y$ in $A$.
Fig.~\ref{fig:fignew8} depicts a total tessellation cover of a graph $G$ and the associated tessellation cover of $A=\text{Tot}(G)$.

Consider the walker located on a vertex $a$ of $G$.
If we apply the operator $H_j$ associated with the color of $a$, the walker hops to the edges incident to $a$ (the edges $ab$ and $ac$).
If we apply an operator associated with the label of an edge incident to $a$, the walker hops to the vertices in the tile of the tessellation of the same label that contains $a$ (the vertices $b$ and $c$). 
The same happens by considering the walker located on a vertex $a$ in $X$. 
If we apply the operator $H_j$ associated with the labels of the edges of $A[K_a]$, the walker hops to the vertices $ab$ and $ac$ of $Y$, and if we apply the operator associated with the label of an edge of $X$ incident to $a$, the walker hops to the vertices $b$ and $c$ of $X$.
Consider the walker located on an edge $ab$ of $G$. If we apply the operator associated with the color of $a$ (or $b$), the walker hops to $a$ (or $b$) and to the edges incident to it. 
The same happens by considering the walker located on a vertex $ab$ in $Y$. 
If we apply the operator associated with the labels of the edges of $A[K_a]$ (or $A[K_b]$), the walker hops to vertices of $K_a$ (or $K_b$).
Otherwise, the walker stays put in both $G$ and $A$.

\begin{figure}

\centering
     \includegraphics[scale=0.2]{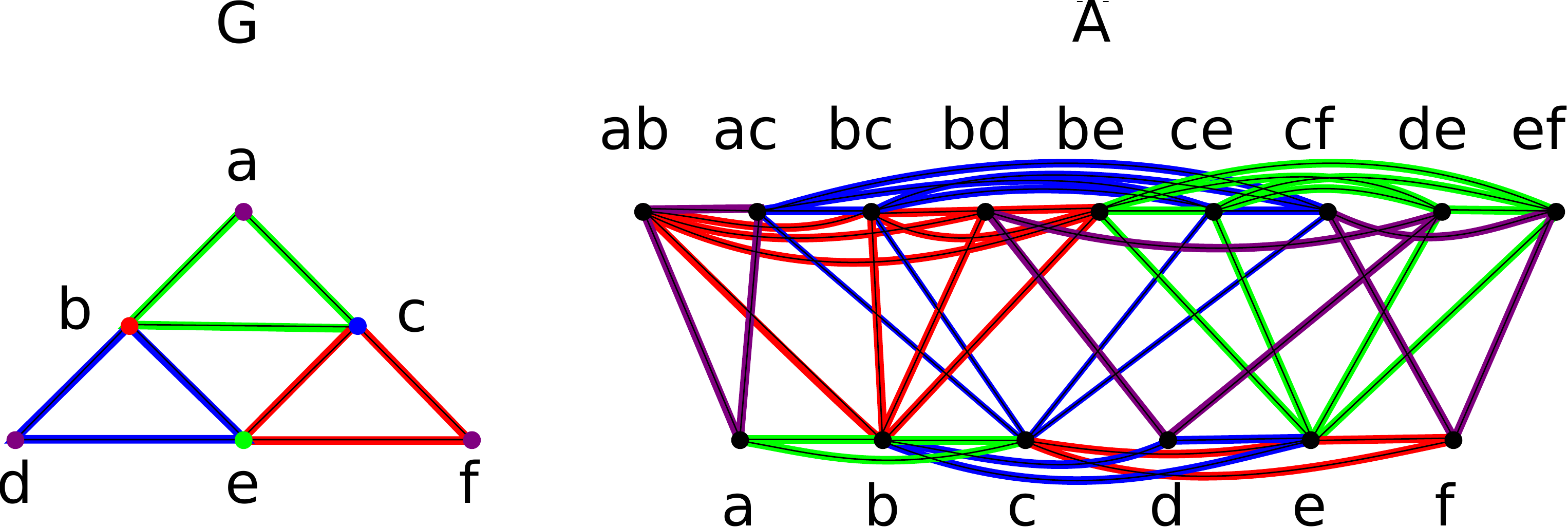}
     \caption{Total tessellation cover of a graph $G$ and the associated tessellation cover of $A$.\label{fig:fignew8}}
\end{figure}


\section{Concluding remarks}\label{sec:five}
We have defined the total tessellation cover on a graph $G$ and have used this concept to define the total staggered quantum walk model. This work strengthens the connection between quantum walk and graph coloring.


We have established examples of graphs for which $T_t(G)$ reaches the bounds of Section~\ref{sec:three}. 
We leave as an open problem to search for graphs with at least $3$ vertices satisfying $T_t(G)=3T(G)$ and $T_t(G) > \chi(G)$. 
Moreover, it would be interesting to define graph classes with $T_t(G)=T(G)=k$ for $k \geq 4$, since for $k=3$ the only such graphs are the odd cycles $C_n$ 
with $n\ \equiv\ 0\ \mod\ 3$.

We have shown that \textsc{$4$-total tessellability} is $\mathcal{NP}$-complete for planar graphs satisfying $is(G)+1=\omega(G) = 4$.
This is important since the hardness of \textsc{$k$-edge colorability} and \textsc{$k$-total colorability} for planar graphs are still open.
On the other hand, we know that planar graphs with large maximum degree have edge and total colorings as small as possible~\cite{artarestaplanar, arttotalplanar}.
We leave as an open problem
to find a threshold for $T_t(G)$ for which all planar graphs are Type~II.

Table~\ref{tab:concluding} summarizes the computational complexities of \textsc{edge-colorability} cf.~\cite{machado2010chromatic}, \textsc{total colorability} cf.~\cite{machadototal}, \textsc{tessellability} cf.~\cite{ArLatin},
and \textsc{total tessellability}. 
These four problems are in $\mathcal{P}$ when restricted to complete graphs, star graphs and trees, whereas for triangle-free graphs, the four problems are $\mathcal{NP}$-complete. 
We leave as an open problem to find  
a graph class for which \textsc{total colorability} is $\mathcal{NP}$-complete and \textsc{total tessellability} is in $\mathcal{P}$.
We have not identified this class because all known $\mathcal{NP}$-completeness proofs of \textsc{total colorability} are restricted to graph classes with $\chi_t(G)=T_t(G)$.


\begin{table}[!h]
{\fontsize{7}{10} \selectfont
    \centering
    \addtolength{\tabcolsep}{-0.8pt}{
    \begin{tabular}
    {|p{1.8cm}|p{.75cm}|p{.75cm}|p{0.3cm}|p{1.8cm}|p{.75cm}|p{.75cm}|p{0.3cm}|p{1.8cm}|p{.75cm}|p{.75cm}|}
    \cline{1-3} \cline{5-7} \cline{9-11}
     & $\chi'(G)$& $T(G)$& & & $\chi'(G)$& $\chi_t(G)$& & & $\chi'(G)$& $T_t(G)$\\
    \cline{1-3} \cline{5-7} \cline{9-11}
     $[2|V(G)|,G^c]$& $\mathcal{P}$ &$\mathcal{NP}$-c & & $G\cup K_{\Delta(G)+1}$, $\Delta$ even & $\mathcal{P}$ &$\mathcal{NP}$-c & & $[2|V(G)|,G^c]$& $\mathcal{P}$ &$\mathcal{NP}$-c  \\
     \cline{1-3} \cline{5-7} \cline{9-11}
     Line~of Bipartite& $\mathcal{NP}$-c & $\mathcal{P}$& & $G\cup K_{\Delta(G)+1}$, $\Delta$ odd & $\mathcal{NP}$-c & $\mathcal{P}$& & Line of Bipartite, $\omega(G) \geq 6$& $\mathcal{NP}$-c & $\mathcal{P}$\\
    \cline{1-3} \cline{5-7} \cline{9-11} 
     \end{tabular}\vspace{0.3cm}
          \begin{tabular}
    {|p{1.8cm}|p{.75cm}|p{.75cm}|p{0.3cm}|p{1.8cm}|p{.75cm}|p{.75cm}|p{0.3cm}|p{1.8cm}|p{.75cm}|p{.75cm}|}
    \cline{1-3} \cline{5-7} \cline{9-11}
  & $T(G)$& $X_t(G)$& & & $T(G)$& $T_t(G)$& & & $\chi_t(G)$& $T_t(G)$\\
    \cline{1-3} \cline{5-7} \cline{9-11}
     Bipartite& $\mathcal{P}$ &$\mathcal{NP}$-c & & Bipartite & $\mathcal{P}$ &$\mathcal{NP}$-c & & $G\cup K_{\Delta(G)+1}$, $\Delta$ odd& $\mathcal{P}$ &$\mathcal{NP}$-c \\
\cline{1-3} \cline{5-7} \cline{9-11}
     $[2|V(G)|,G^c]$& $\mathcal{NP}$-c & $\mathcal{P}$ & & $G\cup K_{3\Delta(G)}$ & $\mathcal{NP}$-c & $\mathcal{P}$& & Open & $\mathcal{NP}$-c & $\mathcal{P}$ \\
    \cline{1-3} \cline{5-7} \cline{9-11}
    \end{tabular}
    
    \caption{Computational complexities of parameters $\chi'(G), \chi_t(G), T(G)$, and $T_t(G)$. \label{tab:concluding}}
    }
}
\end{table}

\bibliographystyle{elsarticle-num}   
\bibliography{tessnumber}


\newpage

\section*{Appendix}

\subsection*{Planar Graphs - Detailed proof of Theorem~\ref{teo:planar}}

The computational complexity of \textsc{total coloring} for planar graphs is an open problem. On the other hand, we show in this section that \textsc{$4$-total tessellability} for planar graphs is $\mathcal{NP}$-complete.
Since in this proof we use a generic graph $G$ such that $is(G)+1=\omega(G)=4$, we also prove that deciding whether a graph has both $T(G)=is(G)+1$ and $T(G)=\omega(G)$ is $\mathcal{NP}$-complete even if restricted to planar graphs.


\begin{lemma}
\label{lem:K4}
Let $G$ be the graph with $V(G) = \{a_1, b_1, c_1, d_1\}\cup\{a_2, b_2, c_2, d_2\}\cup\{a_3, b_3, c_3, d_3\}$, where $\{a_1, b_1, c_1, d_1\}$ is a maximal clique and $\{a_1, a_2, a_3\}$, $\{b_1, b_2, b_3\}$, $\{c_1, c_2, c_3\}$, and $\{d_1, d_2, d_3\}$  are triangles.
Any total tessellation cover of $G$ with four labels has the following property: The edges of three triangles are tiles on a same tessellation and the edges of the remaining triangle are a tile on a different tessellation. 
\end{lemma}
\begin{proof}
The proof follows after analyzing all possibilities of total tessellation covers with four labels. 
\qed
\end{proof}

\begin{figure}
\centering
     \includegraphics[scale=0.45]{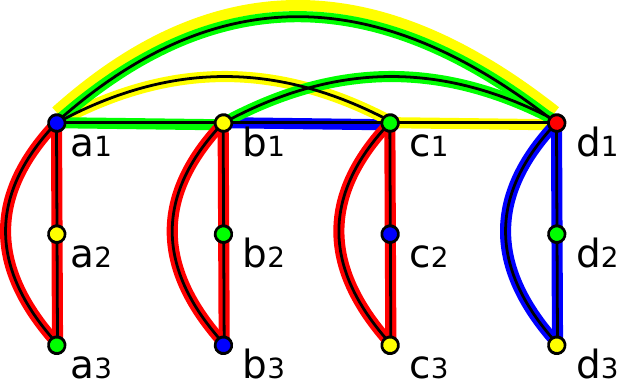}
     \caption{A $4$-total tessellation cover of the graph of Lemma~\ref{lem:K4}.\label{fig:fignew1}}
\end{figure}

Fig.~\ref{fig:fignew1} shows an example of a total tessellation cover of graph $G$ described in Lemma~\ref{lem:K4}. The edges of three triangles must have the same color (red) and the fourth triangle must have a different color (blue).
In the first gadget of the Fig.~\ref{fig:fignew2}, the two external triangles' edges must receive the same label, since the two internal triangles shares a vertex and they must receive different labels.
In the second gadget, since the internal triangles' edges must receive the same label the external triangle' edges must receive different labels.

\begin{figure}
\centering
     \includegraphics[scale=0.3]{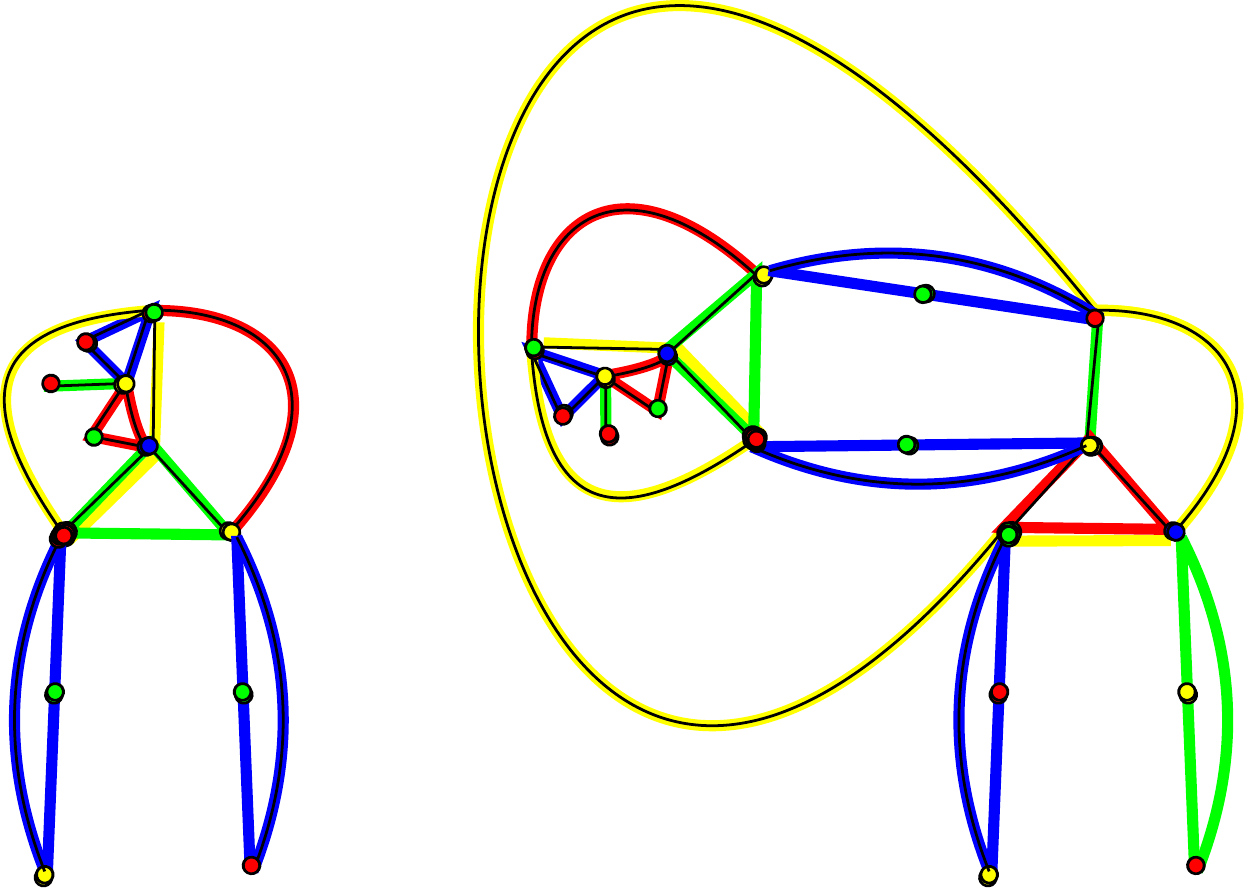}
     \caption{Equal Gadget: edges of its two external triangles are covered by $3$-tiles of a same tessellation in a $4$-total tessellation cover. NotEqual Gadget: edges of its two external triangle are covered by $3$-tiles of different tessellations in a $4$-total tessellation cover.\label{fig:fignew2}}
\end{figure}

\begin{figure}
\centering
     \includegraphics[scale=0.3]{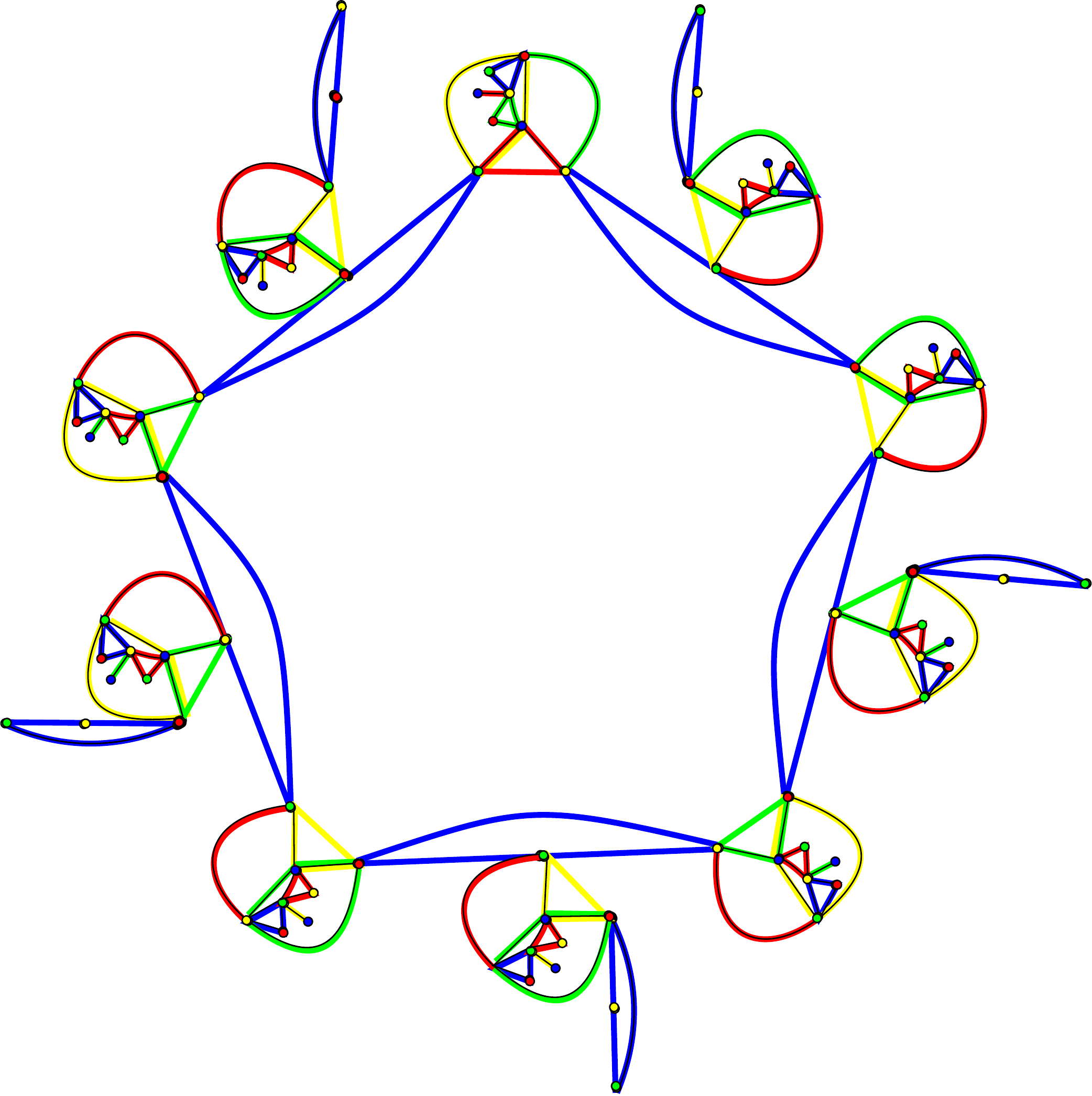}
     \caption{Duplicator Gadget: it forces the five external triangles' edges to have the same label. Moreover, if the label of the triangle' edges is $a$, then its vertices have the next two consecutive labels $a+1$ and $a+2$ modulo 4 available to the vertices of four of these five triangles in a $4$-total tessellation cover.\label{fig:fignew3}}
\end{figure}

\begin{figure}
\centering
     \includegraphics[scale=0.30]{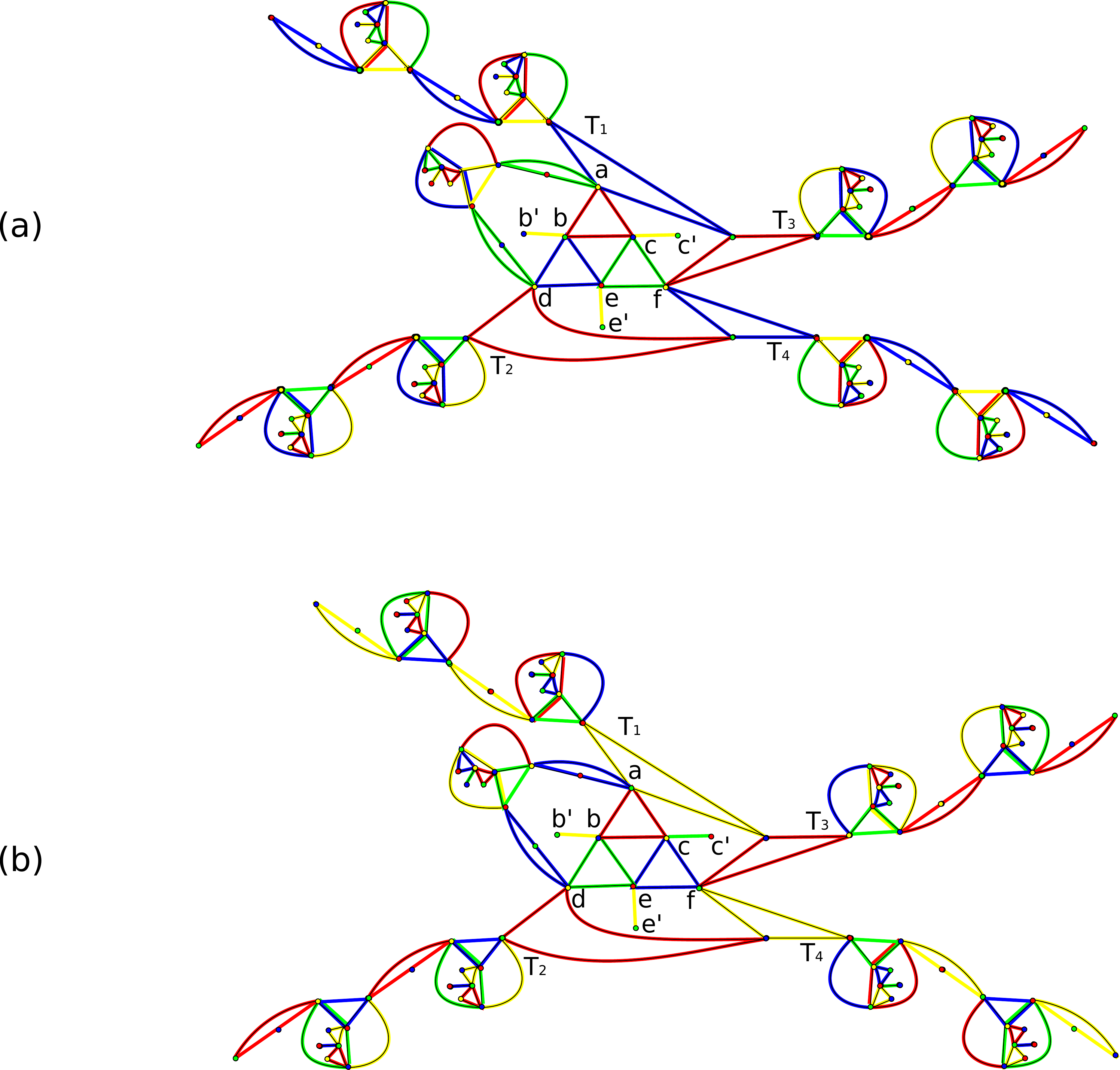}
     \caption{Shifter Gadget: it shifts two crossing tiles in different tessellations such that the tiles get from one side to the other side without crossing edges and maintaining their $3$-tiles tessellations in a $4$-total tessellation cover.\label{fig:fignew4}}
\end{figure}

\begin{lemma}
\label{lem:swift} 
Any total tessellation cover of the graph $G$ of Fig.~\ref{fig:fignew4}  with four labels has the following property: The triangles $T_1$ and $T_4$ are tiles in a same tessellation, and the triangles $T_2$ and $T_3$ are tiles in a same tessellation, which is different from the tessellation that contains $T_1$ and $T_4$.
\end{lemma}
\begin{proof}
Since there are three maximal cliques incident to the vertices $b$ (resp. $c$ and $e$), they are three tiles on different tessellations.
Therefore, the three triangles of the Hajos subgraph of $G$ are tiles on different tessellations.
The Equal Gadget has its tiles incident to the vertices $a$ and $d$ on a same tessellation.

Assume that the color of the vertex $a$ is equal to the color of the vertex $d$.
Since the color of vertex $a$ is the same of the color of vertex $d$ and the tiles of the Equal Gadget have a same tessellation different from the label of the color of $a$ and $d$, it implies that $T_1$ is a tile on the same tessellation of the triangle with vertices $\{b, d, e\}$ and that $T_2$ is a tile on the same tessellation of the triangle with vertices $\{a, b, c\}$.
Note that the colors of the vertices $a$ (resp. $d$) and $b$ are different and that they are also different from the labels of $T_1$ and $T_2$. 
Therefore the color of the vertices $c$ and $e$ must be the same labels of the tessellations of $T_1$ and $T_2$.
Now, the triangle $\{c, e, f\}$ and the vertex $f$ must receive two labels different from the labels used by the triangles $\{a, b, c\}$ and $\{b, d, e\}$.
This implies that the triangles $T_3$ and $T_4$ are tiles with the same labels of the triangles $T_1$ and $T_2$.
Since $T_1$ share a vertex with $T_3$ and $T_2$ share a vertex with $T_4$, the triangles $T_1$ and $T_4$  are tiles on a same tessellation and the triangles $T_2$ and $T_4$ are also tiles on a same tessellation different from the tessellation of $T_1$ and $T_3$.
Therefore, if there is a total tessellation cover of $G$ with $4$ labels, the proof of the theorem holds.
A total tessellation cover of $G$ with $4$ labels is depicted in Fig.~\ref{fig:fignew4} (a).
Note that we obtain total tessellation covers of $G$ with all possible combinations of two distinct labels of the four labels for $T_1$ and $T_2$ by replacing the color classes of $G$ by the desired labels.

Assume that the color of the vertex $a$ is different from the color of the vertex $d$.
For the sake of contradiction assume we use only two labels to the colors of $a$, $d$ and the tiles of the triangles $\{a,b,c\}$ and $\{b,d,e\}$.
This implies that $b$ receives a third label different from these two, and that there is only one available label to the color of the vertices $c$ and $e$, a contradiction.
We also cannot use four different labels to the vertices $a$, $d$ and the triangles $\{a,b,c\}$ and $\{b,d,e\}$ or there would be no available label to the triangles of the Equal Gadget.
Therefore, we have three different labels used in the colors of the vertices $a$, $d$ and the labels of the tiles of the triangles $\{a,b,c\}$ and $\{b,d,e\}$.
This implies that the color of the vertex $b$ and the tiles of the Equal Gadget receive the same label.
The color of the vertices $a$, $b$, and the label of the tile of the triangle $\{a,b,c\}$ are different from the labels of $T_1$.
This implies that the color of $c$ is equal to the label of the tile of $T_1$.
The same holds for the vertex $e$ and the label of the tile of $T_2$.
Now the color of the vertex $f$ and the label of the tile of the triangle $\{c,e,f\}$ must be different from the colors of $c$ and $e$ (i.e., the label of the tiles of $T_1$ and $T_2$).
This implies that the label of the tiles $T_3$ and $T_4$ are the same labels of the tiles $T_1$ and $T_2$.
Since $T_1$ share a vertex with $T_3$ and $T_2$ share a vertex with $T_4$, the triangles $T_1$ and $T_4$  are tiles on a same tessellation and the triangles $T_2$ and $T_4$ are also tiles on a same tessellation different from the tessellation of $T_1$ and $T_3$.
Therefore, if there is a total tessellation cover of $G$ with $4$ labels, the proof of the theorem holds.
A total tessellation cover of $G$ with $4$ labels is depicted in Fig.~\ref{fig:fignew4} (b).
Note that we obtain total tessellation covers of $G$ with all possible combinations of two distinct labels of the four labels for $T_1$ and $T_2$ by replacing the color classes of $G$ by the desired labels.
\qed
\end{proof}

\begin{customthm}{3}
\textsc{$4$-total tessellability} is $\mathcal{NP}$-complete for planar graphs.
\end{customthm}
\begin{proof}
Let $G$ be an instance of \textsc{$3$-colorability} of planar graphs with degree at most four~\cite{ArGarey}.
Add a universal vertex $u$ to $G$ so that $G\vee\{u\}$ has a $4$-coloring if and only if $G$ has a $3$-coloring.
We create a planar graph $H$ from $G \vee \{u\}$ as follows.
We replace each vertex of $G \vee \{u\}$ by a Duplicator Gadget with the degree of the vertex duplication in $H$.
We replace each edge of $G\vee\{u\}$ by a NotEqual Gadget connecting the related triangles of the Duplicator Gadgets of the endpoints of the edge in $H$.
The only crossing edges in $H$ are from the triangles of the universal vertex and the triangles of the other Duplicator Gadget that have labels different from the one of the universal Gadget.
We  replace these crossing tiles with the Shifter Gadget.

We claim that the resulting planar graph $H$ has a $4$-total tessellation cover if and only if the graph $G$ has a $3$-coloring.

Consider a $4$-total tessellation cover of $H$. If two vertices are adjacent in $G\vee\{u\}$, then the NotEqual Gadget forces the tiles of the external triangles of the respective Duplicator Gadgets of these two adjacent vertices to be on different tessellations.
Therefore, we obtain a $4$-coloring of $G\vee\{u\}$ by assigning the color of a vertex as the label of the tile of the external triangles of the Duplicator Gadget related to that vertex.

Consider a $4$-coloring $f$ of $G \vee \{u\}$.
We obtain a $4$-total tessellation cover of $H$ as follows. 
Assign each tile of the external triangles of the Duplicator Gadget to the tessellation related to the color the vertex received in $f$.
Label the remaining vertices and edges as described in Figure~\ref{fig:fignew3} by rotating the color classes labels to obtain the desired label.

Since we obtain the total tessellation cover of the Duplicator Gadget by rotating the color classes, we have that the label of a external triangle and a vertex of the degree two is related to consecutive colors.
Therefore, if the label of the tile of the external triangle is $1$ (resp. $2$, $3$, and $4$), then there are two vertex of degree two in this external triangle with colors $2$ and $3$ (resp. $3$ and $4$, $4$ and $1$, $1$ and $2$). 
Now, for any two different tessellations of the tiles of the external triangles, we select one vertex of degree two of each so that we do not use all four labels in these two vertex and in the two tiles of the external triangles.
By Lemma~\ref{lem:K4}, there is a total tessellation cover with four labels of the NotEqual Gadget if we do not use all four labels on the two tiles of its external triangles and the two vertices of the $K_4$ of that external triangles.

We obtain a total tessellation cover with $4$ labels of the Shifter Gadgets as described in Lemma~\ref{lem:swift}.
Note that, as depicted in Fig.~\ref{fig:fignew4}, the two consecutive Equal Gadgets connected to the external triangles $T_1$ (resp. $T_2$, $T_3$, and $T_4$) allow us to assign colors to the vertices of the Shifter Gadgets so that the vertices of the last of their external triangles have the same colors of the vertices of the external triangles of the Duplicator Gadgets that they are related.
\qed
\end{proof}

\end{document}